\documentclass[]{tCON2e}
\usepackage{natbib}
\usepackage{amsmath,amsfonts,amssymb}
\usepackage{graphicx}      

\newcommand{\be}{\begin{equation}}
\newcommand{\ee}{\end{equation}}
\newcommand{\bd}{\begin{displaymath}}
\newcommand{\ed}{\end{displaymath}}
\newcommand{\bea}{\begin{eqnarray}}
\newcommand{\eea}{\end{eqnarray}}

\newcommand{\R}{\mathbb{R}}
\newcommand{\C}{\mathbb{C}}
\newcommand{\Z}{\mathbb{Z}}

\newcommand{\Cs}{\mathbb{C}_{\sigma_0}}

\newcommand{\w}{\omega}
\newcommand{\so}{\sigma_0}
\newcommand{\dwz}{\Delta\w_{zk}}
\newcommand{\dwp}{\Delta\w_{pi}}
\newcommand{\dsz}{\Delta \sigma_{zk}}
\newcommand{\dsp}{\Delta \sigma_{pi}}

\begin{document}

\markboth{S. Gumussoy}{Relative Stability Analysis of Closed-Loop SISO Dead-Time Systems}

\articletype{}

\title{Relative Stability Analysis of Closed-Loop SISO Dead-Time Systems: Non-imaginary Axis Case}

\author{Suat Gumussoy$\dagger$$^{\ast}$\thanks{$^\ast$Corresponding author. Email: suat.gumussoy@cs.kuleuven.be
\vspace{6pt}}\\\vspace{6pt}{$\dagger$\em{Department of Computer Science, K. U. Leuven, \\
        Celestijnenlaan 200A, 3001, Heverlee, Belgium}};
\\\vspace{6pt}\received{\today} }

\maketitle
\begin{abstract}
We present a numerical method to analyze the relative stability of closed-loop single-input-single-output (SISO) dead-time systems on a given left complex half-plane for all positive delays. The well-known boundary crossing method for the imaginary axis is extended to a given vertical line stability boundary in the complex plane for these type of systems. The method allows to compute the characteristic roots crossing the relative stability boundary and their corresponding delays up to a maximum predefined delay. Based on this method, we analyze the relative stability of the closed-loop system for all positive delays. Both numerical methods are effective for high-order SISO dead-time systems.

\end{abstract}

\section{Introduction}
The stability of time-delay systems with commensurate delays is analyzed with respect to the time-delay using different approaches in the literature, \cite{NiculescuBook, GuBook, WimBook}. Many approaches are based on the continuity of characteristics roots with respect to the time-delay for retarded time-delay systems and neutral time-delay systems with finitely many characteristics roots on the right complex half-plane including the imaginary axis, \cite{WimBook}. These approaches are based on the computation of characteristic roots crossing the stability boundary, the imaginary axis. Various methods are available for this computation such as two variable-based approaches \cite{GuBook, OlgacTAC02, OlgacTAC04}; matrix pencil approaches \cite{ChenSCL1995, FuTAC06}; geometric approaches \cite{Gu05}. The stability of time-delay systems with respect to the time-delay is analysed using the information on the computed boundary crossing roots.

There are other methods for robust stability analysis of time-delay systems in parameter and delay-parameter spaces,  \cite{KharitonovSurvey98, RichardSurvey03, SipahiSICON09, FazeliniaTAC2007, SipahiDSMC09, JarlebringCAM09} (see references therein).

In practical implementations, it is desirable to consider the \emph{relative stability boundary}, a \emph{given} vertical line with a negative real part and parallel to the imaginary axis, \cite{Bourles87, Niculescu94}. A system is \emph{relative stable} if all characteristic roots of the system lie left side of the relative stability boundary. The distance of the relative stability boundary from the imaginary axis is considered as a stability margin and determines the settling time of the system, \cite{Filipovic2002}. This guarantees that the system is stable against parameter uncertainties and disturbances and satisfies the settling time requirements.

One approach to check the relative stability of the time-delay systems with \emph{constant fixed delays} is to compute the right-most characteristic root by approximating the spectrum of the time-delay system using spectral methods, \cite{EngelBorghs2000, EngelBorghs2002, BredaSISC2005}. The real part of the right-most characteristic root is computed by solving two nonlinear equations, from magnitude and phase conditions, by a constrained numerical search in \cite{Filipovic2002}. Then the relative stability of the system is determined by the location of its right-most characteristic root in the complex plane. Note that the relative stability analysis for a delay interval by these approaches requires the computation of characteristic roots for each delay value of a discretized set of the delay interval. This is numerically expensive and the fast dynamics of a characteristic root between two consecutive delay values may not be detected.

We propose a numerical method to analyze the \emph{relative stability} of closed-loop SISO dead-time systems with respect to its \emph{delay}. To the author's best knowledge, this is the first stability analysis method for time-delay systems on the stability boundary different than the imaginary axis. SISO dead-time systems are frequently encountered in mechanical systems and chemical process applications such as the feedback control a delayed resonator, \cite{Filipovic2002} and a continuous stirred tank reactor, \cite{Huang96}; in biosciences applications such as the human respiration model and the population periodic cycles, \cite{WimBook}. Therefore the relative stability analysis of closed-loop SISO dead-time systems is important for practical applications.

This paper presents two main contributions. First is computation of delay intervals and the number of characteristic roots to the right of a given relative stability boundary on these intervals for a given upper bound delay. The second contribution is the relative stability analysis of closed-loop SISO dead-time systems for \emph{all delays}. It is shown that this analysis requires the computation of finitely many boundary crossing roots because the crossing direction of asymptotic characteristic roots is invariant. Then the relative stability analysis follows from the number of characteristic roots on each delay interval similar to the approach in \cite{OlgacTAC02, OlgacTAC04}.

The proposed method is numerically exact and detects all characteristic roots crossing the given relative stability boundary for the given upper bound delay. The effectiveness of the relative stability analysis is tested on various high-order SISO dead-time systems and some benchmark results are given.

The rest of this paper is organized as follows. Section \ref{sec:problem} formulates the considered problems. In Section~\ref{sec:intervals}, we compute the intervals on the relative stability boundary such that the time-delay values of the boundary crossing roots are non-negative and their crossing directions are invariant. Two main contributions are given in Section \ref{sec:mainres}. Section~\ref{sec:examples} is devoted to
the numerical examples. In Section~\ref{sec:concl} some concluding remarks are presented.

\textbf{Notation:} \\
The notations in this paper are standard and given below.
\begin{tabbing}
  \= $\C, \R, \Z \quad\quad $ \=: sets of complex, real and integer numbers, \\
  \> $\Re(u), \Im(u)$ \>: real and imaginary parts of a complex number $u$, \\
  \> $|u|, \angle u$ \>: the magnitude and phase of a complex number $u$, \\
  \> $\lfloor u \rfloor, \lceil u \rceil$ \>: the next smallest and largest integers close to a real number $u$, \\
  \> $\textrm{sign}(u)$ \>: returns $+1, -1, 0$ given a real number $u$ for $u>0, u<0, u=0$ respectively, \\
  \> $\frac{\partial f}{\partial h}, \frac{\partial f}{\partial s}$ \>: partial derivatives of the function $f(s,h)$ with respect to its parameters $h$ and $s$.
\end{tabbing}

\section{Problem Formulation} \label{sec:problem}

The transfer function of the closed-loop SISO dead-time system is
\be \label{eq:CL}
G_{cl}(s,h):=(1+G(s)e^{-hs})^{-1},\; h>0,\;h\in\R.
\ee The plant $G$ is a proper SISO system and it has transfer function representation
\be \label{eq:tfG}
G(s)=k_G\frac{\prod_{k=1}^m s-(\sigma_{zk}+j\omega_{zk})}
{\prod_{i=1}^n s-(\sigma_{pi}+j\omega_{pi})}
\ee where $k_G\in\R$, $\sigma_{zk}+j\omega_{zk}$ $k=1,\ldots,m$, $\sigma_{pi}+j\omega_{pi}$ $i=1,\ldots,n$ are the system gain, zeros and poles of $G$ respectively. Note that a real zero or pole of the transfer function (\ref{eq:tfG}) is represented by the real part equal to the value of the system zero or pole and the imaginary part  equal to zero and the complex zeros and poles appear in complex-conjugate pairs.

The \emph{relative stability boundary} is the \emph{given} vertical line parallel to the imaginary axis, $\Re(s)=\sigma_0<0$,  $\sigma_0\in\R$ and its value depends on the design requirements.

We define the \emph{complementary stability region} $\Cs$ as
\be \label{eq:Cs}
\Cs =\left\{s\in \mathbb{C}: \Re(s)\geq\sigma_0\right\}.
\ee The closed-loop system $G_{cl}$ is stable on the complex left half-plane, $\R(s)<\sigma_0$ when there is no characteristic root inside $\Cs$.

The \emph{crossing direction} of a characteristic root $s$ crossing $\Re(s)=\so$ at the time-delay $h$ determines whether a characteristic root enters into or leaves the region $\Cs$, \cite{WimBook} and it is defined as

\be \label{eq:RT}
\mathcal{CD}(s,h):=\textrm{sign}\left(\Re{\left(\left.-\left(\frac{\partial f}{\partial h}\right)\left(\frac{\partial f}{\partial s}\right)^{-1}\right|_{f(s,h)=0}\right)}\right)\ \textrm{where}\ f(s,h)=1+G(s)e^{-hs}.
\ee

We consider the following two problems for a closed-loop SISO dead-time system $G_{cl}$:
\begin{enumerate}
  \item Compute delay intervals and the number of characteristic roots of $G_{cl}$ inside $\Cs$ at each delay interval for $h\in[0,h_{\max}]$,
  \item Find delay intervals where $G_{cl}$ is stable on the region $\Re(s)<\so$ for $h\in[0,\infty)$.
\end{enumerate}

Both proposed methods are based on computing \emph{critical delays} in an increasing order when characteristic roots of $G_{cl}$ cross the relative stability boundary $\R(s)=\so$. This computation requires finding the boundary crossing roots $s$ and their critical delays $h$ satisfying
\be \label{eq:complex}
1+G(s)e^{-h s}=0\quad\rm{and}\quad \Re(s)=\so.
\ee

The relative stability boundary can be mapped to the imaginary axis by $\tilde{s}=s-\so$ and the problem is transformed into
\be
1+G(\tilde{s}+\so)e^{-h\tilde{s}} e^{-h \so}=0\quad\rm{and}\quad \Re(\tilde{s})=0.
\ee

This problem is difficult compared to the standard problem on the imaginary axis where $\so=0$. The extra term $e^{-h\so}$ does not allow to use standard delay elimination techniques and to reduce the problem into a finite dimensional problem, \cite{WimBook}. We propose a new method to compute boundary crossing roots, their critical delays on the relative stability boundary and this method allows us to analyze the relative stability of closed-loop SISO dead-time systems.

In Section \ref{sec:intervals}, we find \emph{feasible intervals with invariant crossing direction} on the relative stability boundary where time-delays in (\ref{eq:complex}) are non-negative for all points inside these intervals and their crossing directions are invariant. Then we compute boundary crossing roots and their critical delays on these intervals in Section~\ref{sec:mainres}.

\section{Computation of Feasible Intervals with Invariant Crossing Direction on the Relative Stability Boundary} \label{sec:intervals}

We compute \emph{feasible intervals} and \emph{intervals with invariant crossing direction} on the relative stability boundary in Sections \ref{sec:whintvals} and \ref{sec:w1intvals} respectively. Feasible intervals with invariant crossing direction follow from the intersection of two sets of intervals.

\subsection{Computation of Feasible Intervals} \label{sec:whintvals}
By the magnitude equation of (\ref{eq:complex}), a boundary crossing root of $G_{cl}$, $\so+j\w$, satisfies
\be \label{eq:hw}
H(\omega):=\frac{\ln\left|G(\sigma_0+j\omega)\right|}{\sigma_0}\geq0.
\ee The \emph{feasible intervals} are the intervals on the relative stability boundary, $\Re(s)=\so$ and $\w\in[0,\infty)$ such that the inequality (\ref{eq:hw}) holds. The computation of these intervals is based on finding the subinterval satisfying (\ref{eq:hw}) on each interval that $H(\w)$ is monotonic (increasing or decreasing). The following lemma gives some properties of $H(\w)$.

\begin{lemma} \label{lem:propH}
Let $G$ have a transfer function representation (\ref{eq:tfG}). Assume that $G$ has no poles or zeros on the relative stability boundary. The functions $H(\w)$, $H'(\w)$ and $H''(\w)$ are continuous and the non-negative zeros of $H'(\w)$ are the non-negative real roots of the polynomial,
\be \label{eq:polyHp}
\Gamma_{p}(\w)\sum_{k=1}^m(\w-\w_{zk})\Gamma_z^k(\w)-
\Gamma_z(\w)\sum_{i=1}^n(\w-\w_{pi})\Gamma_p^i(\w)
\ee
where $\dsz=(\so-\sigma_{zk})$, $\dwz=(\w-\w_{zk})$, $\gamma_{zk}(\w)=\dsz^2+\dwz^2$, $\Gamma^k_z(\w)=\prod_{\substack{k_1=1\\k_1\neq k}}^m\gamma_{zk}(\w)$ for $k=1,\ldots,m$, $\dsp=(\so-\sigma_{pi})$, $\dwp=(\w-\w_{pi})$, $\gamma_{pi}(\w)=\dsp^2+\dwp^2$, $\Gamma^i_p(\w)=\prod_{\substack{i_1=1\\i_1\neq k}}^n\gamma_{pi}(\w)$ for $i=1,\ldots,n$, $\Gamma_z(\w)=\prod_{k=1}^m\gamma_{zk}(\w)$, $\Gamma_p(\w)=\prod_{i=1}^n\gamma_{pi}(\w)$.
\end{lemma}

\begin{proof}
Substituting the transfer function of $G$ (\ref{eq:tfG}) into
(\ref{eq:hw}), the function $H(\w)$ can be written as, \be \label{eq:H}
H(\omega)=\frac{\ln|k_G|}{\sigma_0}+\frac{1}{2\sigma_0}\left(\sum_{k=1}^m
\ln\gamma_{zk}(\omega)-\sum_{i=1}^n
\ln\gamma_{pi}(\omega)\right). \ee
By taking the first and the second derivatives of $H(\w)$ (\ref{eq:H}), we get
\bea
H'(\w)&=&\frac{1}{\sigma_0}\left(\sum_{k=1}^m \frac{\dwz}{\gamma_{zk}(\omega)}-\sum_{i=1}^n \frac{\dwp}{\gamma_{pi}(\omega)} \right), \label{eq:Hp} \\
H''(\w)&=&\frac{1}{\sigma_0}\left(\sum_{k=1}^m \frac{\dwz^2-\dsz^2}{\gamma_{zk}^2(\omega)}
-\sum_{i=1}^n \frac{\dwp^2-\dsp^2}{\gamma_{pi}^2(\omega)} \right). \label{eq:Hpp}
\eea

Note that $H(\w)$, $H'(\w)$ and $H''(\w)$ are continuous except the points where $\gamma_{zk}(\w)$ or $\gamma_{pi}(\w)$ are equal to zero. These points are the poles or zeros of $G$ on $\Re(s)=\so$. The continuity results in Lemma \ref{lem:propH} follow. The polynomial (\ref{eq:polyHp}) is the numerator of the function $H'(\w)$ in (\ref{eq:Hp}) and the result follows.
\end{proof}

\begin{corollary} \label{cor:extremumpts}
Assume that $G$ has no poles or zeros on the relative stability boundary. Then function $H(\w)$ is monotonic on the intervals whose boundary points (without multiplicity) are the non-negative zeros of $H'(\w)$ (computed by Lemma \ref{lem:propH}), $0$ and $\infty$.
\end{corollary}
\begin{proof}
By Lemma \ref{lem:propH}, the function $H(\w)$ is continuous since $\so$ is chosen such that there are no poles or zeros of $G$ on $\Re(s)=\so$. Therefore it is monotonic inside the intervals determined by its extremum points and the end points of its domain ($0$ and $\infty$).
\end{proof}

Because $H(\w)$ is monotonic on each interval in Corollary~\ref{cor:extremumpts}, we can find the subinterval where $H(\w)$ takes non-negative values satisfying (\ref{eq:hw}). For example, if the values of $H(\w)$ at the interval end points are non-negative, then $H(\w)$ is non-negative in this interval due to the monotonicity of $H(\w)$. Similarly if the values of $H(\w)$ at the interval end points have different signs, we can find the point where $H(\w)$ is equal to zero and choose the subinterval satisfying (\ref{eq:hw}). Based on this idea, feasible intervals $I_h=\cup_{i=1}^{n_h}I_h^i$ are computed by the following algorithm.
\begin{algorithm} \label{alg:H}
By Corollary \ref{cor:extremumpts}, compute all the intervals that $H(\w)$ is monotonic. \\
For each given interval $[\w_l^i,\w_u^i]$,
\begin{enumerate}
  \item $H(w_l^i)\geq0$ and $H(w_u^i)\geq0\ \Rightarrow\ I_h^i=[\w_l^i,\w_u^i]$,
  \item $H(w_l^i)<0$ and $H(w_u^i)<0 \Rightarrow\ I_h^i=\varnothing$,
  \item $H(w_l^i)H(w_u^i)<0$, find $\w_0$ such that $H(\w_0)=0$ by a bisection search, then
  \begin{itemize}
    \item $I_h^i=[\w_0,w_u^i]$ if $H(\w)$ is increasing,
    \item $I_h^i=[w_l^i,\w_0]$ if $H(\w)$ is decreasing.
  \end{itemize}
  \item $H(w_l^i)=0$ and $H(w_u^i)<0 \Rightarrow\ I_h^i=\{w_l^i\}$ or  \\
        $H(w_l^i)<0$ and $H(w_u^i)=0 \Rightarrow\ I_h^i=\{w_u^i\}$.
\end{enumerate}
\end{algorithm}

\begin{remark} The last case $(4)$ in the algorithm occurs when $H(\w)$ has multiple non-negative zeros and the other boundary point is negative. This isolated point is discarded if it is contained in an adjacent interval or considered as a single point otherwise.
\end{remark}

\begin{example} \label{ex:Ih} Given the plant
\be \label{eq:ex1}
G(s)=\frac{2s^2+s+3}{s^3+2s^2+3s+4}
\ee and the relative stability boundary $\Re(s)=-0.1$, $H(\w)$ plot is shown in Figure~\ref{fig:wh}. By Corollary~\ref{cor:extremumpts}, diamond-shaped points are boundary points, non-negative zeros of $H'(\w)$ and end points of the domain (the infinity is not shown). Note that $H(\w)$ is monotonic between diamond-shaped dots. Feasible intervals (shown in bold lines) are computed by Algorithm \ref{alg:H}. The first two diamond-shaped points are positive and the whole interval is the feasible interval $I_h^1$ by the first condition in the algorithm. The next two pairs of the diamond-shaped points have opposite signs. For each pair, the point $\w_0$ is computed by a bisection search (shown as circle-shaped points) and the feasible intervals are found by the third condition in the algorithm. The characteristic roots pass through only these intervals on the relative stability boundary.
\end{example}

\begin{figure}[!h]
   \begin{minipage}[t]{0.45\textwidth}
      \vspace{0pt}
      \includegraphics[width=\linewidth]{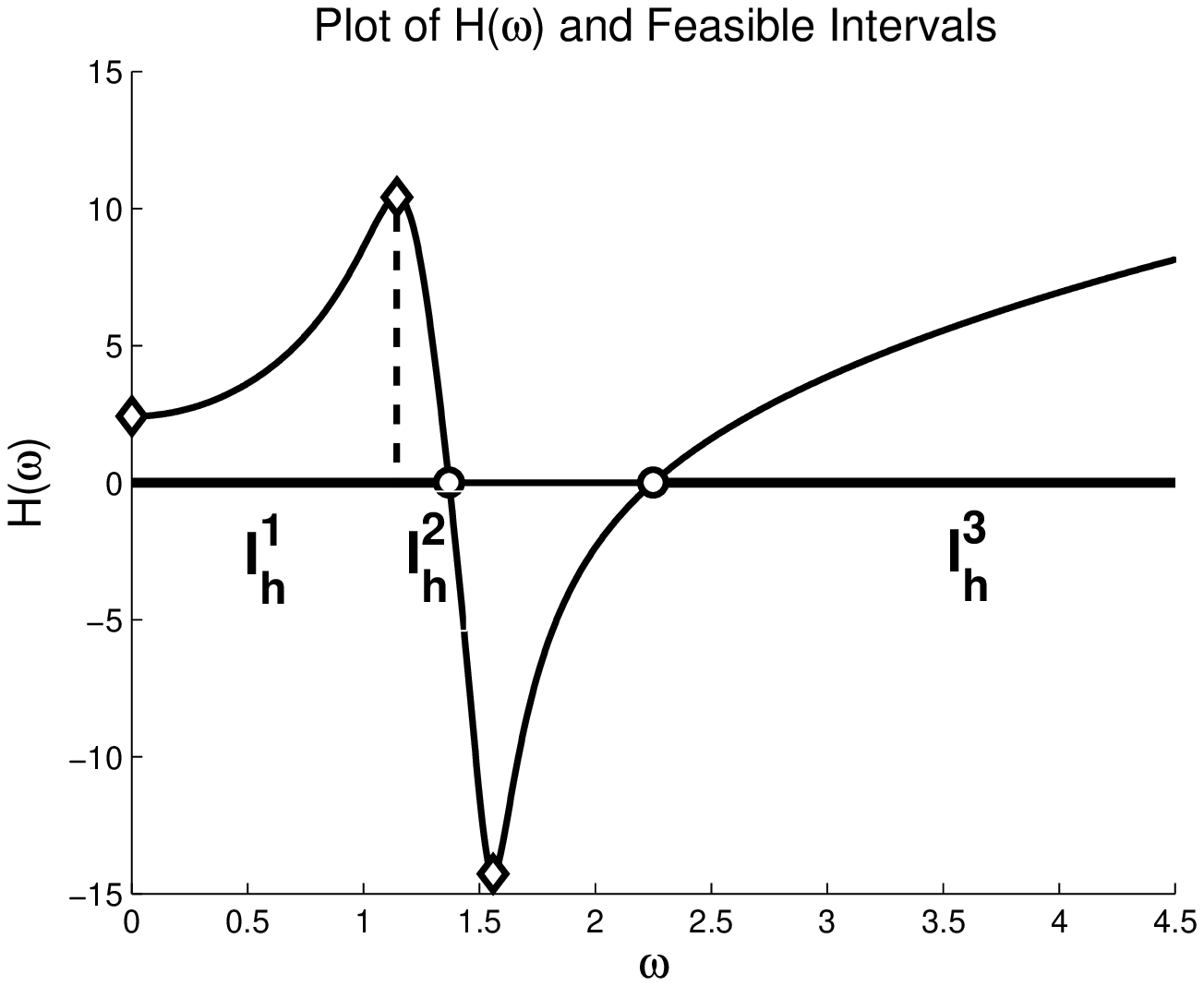}
      \caption{$H(\w)$ plot and feasible intervals, $I_h^1$, $I_h^2$, $I_h^3$.}
      \label{fig:wh}
   \end{minipage}
   \hfill
   \begin{minipage}[t]{0.45\textwidth}
      \vspace{0pt}\raggedright
      \includegraphics[width=\linewidth]{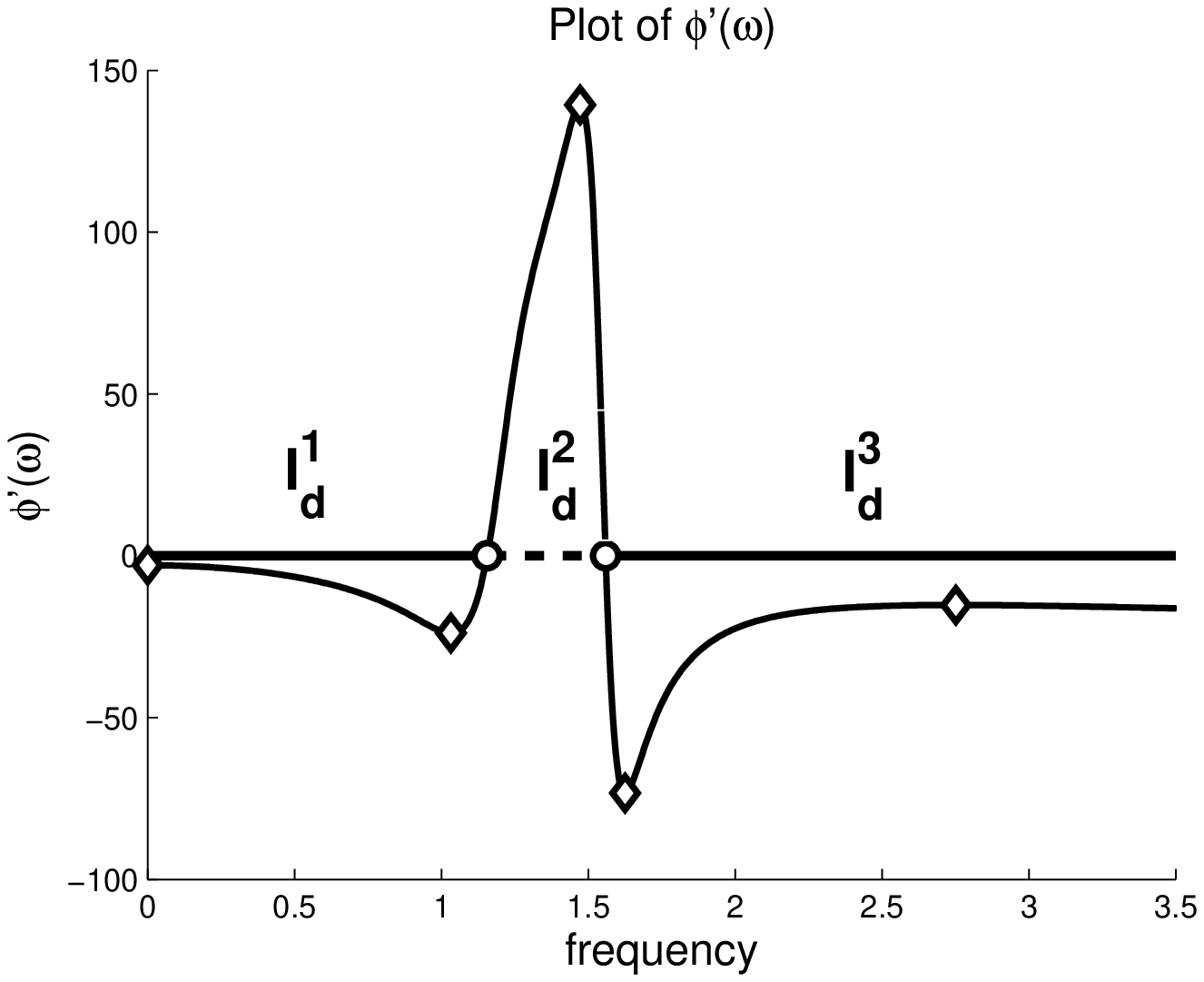}
      \caption{$\phi'(\w)$ plot and the frequency intervals with invariant crossing direction.}
      \label{fig:w1}
   \end{minipage}
\end{figure}

\subsection{Computation of Intervals with Invariant Crossing Direction} \label{sec:w1intvals}

We show that whether a characteristic root of $G_{cl}$ enters into or leaves the complementary stability region depends on \emph{only} the interval crossed on the relative stability boundary. Therefore, the crossing directions of the boundary crossing roots are same if their imaginary parts are inside the same interval on the relative stability boundary.

\begin{theorem} \label{thm:RT}
Let $s_0=\sigma_0+j\w_0$ be a characteristic root of $G_{cl}$ on the relative stability boundary at $h=h_0$.
Then the crossing direction of $s_0$ is,
\be \label{eq:RTthm}
\mathcal{CD}(s_0,h_0)=\textrm{sign}\left(\so \phi'(\w_0)\right).
\ee where
\be
\phi(\w):=\sum_{k=1}^m \tan^{-1}\left(\frac{\dwz}{\dsz}\right)
-\sum_{i=1}^n \tan^{-1}\left(\frac{\dwp}{\dsp}\right) \label{eq:phi}
-\w H(\w) \ \rm{for}\ \w\in[0,\infty).
\ee
\end{theorem}
\begin{proof}
The derivative $\phi(\w)$ (\ref{eq:phi}) at $\w=\w_0$ is
\be \label{eq:phip}
\phi'(\w_0)=\left.\left(\sum_{k=1}^m \frac{\dsz}{\gamma_{zk}(\w)}-
\sum_{i=1}^n \frac{\dsp}{\gamma_{pi}(\w)}-H(\w)-\w H'(\w)\right)\right|_{\w=\w_0}.
\ee
Using the transfer function representation in (\ref{eq:tfG}) and (\ref{eq:Hp},\ref{eq:phip}), we obtain
\be \label{eq:GpG}
G'(s_0)G^{-1}(s_0)-h_0=\phi'(\w_0)+\w_0 H'(\w_0)-j\sigma_0 H'(\w_0).
\ee
By (\ref{eq:RT}) and (\ref{eq:GpG}), the crossing direction of $\so$ at $h=h_0$ is equal to
\bea \label{eq:RTthm1}
\nonumber \mathcal{CD}(s_0,h_0)&=&\textrm{sign}\left(\Re\left(s_0\left(\frac{G'(s_0)}{G(s_0)}-h_0\right)^{-1}\right)\right)=
\textrm{sign}\left(\Re\left(\frac{\so+j\w_0}{\phi'(\w_0)+\w_0 H'(\w_0)-j\sigma_0 H'(\w_0)}\right)\right), \\
\nonumber &=&\textrm{sign}(\so\phi'(\w_0)).
\eea
\end{proof}

The continuity properties of $\phi(\w)$ up to the second derivatives and the computation of the non-negative zeros of $\phi''(\w)$ are given in the following lemma.

\begin{lemma} \label{lem:propPhi}
Assume that $G$ has no poles or zeros on the relative stability boundary. The functions $\phi(\w)$, $\phi'(\w)$ and $\phi''(\w)$ are continuous and the non-negative zeros of $\phi''(\w)$ are the non-negative real roots of the polynomial,
\be \label{eq:p2}
(\Gamma_p(\w))^2\sum_{k=1}^m \Phi_z^k(w)(\Gamma_z^k(\w))^2-
(\Gamma_z(\w))^2\sum_{i=1}^n\Phi_p^i(w)(\Gamma_p^i(\w))^2
\ee
where the functions $\Gamma^k_z(\w)$ for $k=1,\ldots,m$, $\Gamma^i_p(\w)$ for $i=1,\ldots,n$, $\Gamma_z(\w)$ and $\Gamma_p(\w)$ are defined in Lemma \ref{lem:propH} and $\Phi_z^k(\w)=(2\w_{zk}-3\w)\dsz^2+(2\w_{zk}-\w)\dwz^2-2\so\dsz\dwz$ for $k=1,\ldots,m$, $\Phi_p^i(\w)=(2\w_{pi}-3\w)\dsp^2+(2\w_{pi}-\w)\dwp^2-2\so\dsp\dwp$ for $i=1,\ldots,n$.
\end{lemma}

\begin{proof}
The second derivative of the function $\phi(\w)$ is
\be \label{eq:phipp}
\phi''(\w)= -2 \sum_{k=1}^m \frac{\dsz\dwz}{\gamma_{zk}^2 (\w)}+2 \sum_{i=1}^n \frac{\dsp\dwp}{\gamma_{pi}^2(\w)}-2 H'(\w) - \w H''(\w).
\ee

By Lemma \ref{lem:propH}, the functions $H(\w)$, $H'(\w)$ and $H''(\w)$ are continuous if there are no poles or zeros of $G$ on $\Re(s)=\so$. The terms in the functions $\phi(\w)$ (\ref{eq:phi}), $\phi'(\w)$ (\ref{eq:phip}) and $\phi''(\w)$ (\ref{eq:phipp}) are continuous except the points that the condition holds. The continuity results follow. The polynomial (\ref{eq:p2}) is the numerator of the function $\phi''(\w)$ (\ref{eq:phipp}) and the result follows.
\end{proof}

\begin{corollary} \label{cor:invariantpts}
Assume that $G$ has no poles or zeros on the relative stability boundary. Then the crossing direction of characteristic roots of $G_{cl}$ is invariant inside the intervals whose boundary points (without multiplicity) are the non-negative zeros of $\phi'(\w)$, $0$ and $\infty$.
\end{corollary}
\begin{proof}
By Lemma \ref{lem:propPhi} the function $\phi(\w)$ is continuous, the monotonicity of $\phi(\w)$ changes only at the points where $\phi'(\w)=0$. The assertion follows from Theorem \ref{thm:RT}.
\end{proof}

Based on the monotonicity of $\phi'(\w)$ on the intervals whose boundary points are the non-negative zeros of $\phi''(\w)$, $0$ and $\infty$, the non-negative zeros of $\phi'(\w)$ are computed by a bisection search. Then the boundary points of the intervals with invariant crossing direction $I_d=\cup_{i=1}^{n_d} I_d^i$ are the non-negative zeros of $\phi'(\w)$, $0$ and $\infty$. The algorithm is given below.
\begin{algorithm} \label{alg:Phi}
Compute the set $\{w_{2i}\}_{i=1}^{n_2}$ whose elements are the non-negative zeros of $\phi''(\w)$ calculated by Lemma \ref{lem:propPhi} , $0$ and $\infty$ in ascending order. \\
For each pair $\w_{2i},\w_{2(i+1)}$
\begin{enumerate}
  \item $\phi'(\w_{2i})=0$, set $\w_1=\w_{2i}$,
  \item $\phi'(\w_{2(i+1)})=0$, set $\w_1=\w_{2(i+1)}$,
  \item $\phi'(\w_{2i})\phi'(\w_{2(i+1)})<0$, find $\w_1$ such that $\phi'(\w_1)=0$ by a bisection search,
  \item add $\w_1$ to the set of non-negative zeros of $\phi'(\w)$.
\end{enumerate}
The boundary points of the intervals $I_d$ are the set of non-negative zeros of $\phi'(\w)$, $0$ and $\infty$ by Corollary \ref{cor:invariantpts}.
\end{algorithm}

\begin{example} \label{ex:Id}
For the same example, $\phi'(\w)$ plot is shown in Figure~\ref{fig:w1}. Diamond-shaped points are the non-negative zeros of $\phi''(\w)$ computed by Lemma~\ref{lem:propPhi}. Since there are two sign changes in diamond-shaped points, two zeros of $\phi'(\w)$ are computed by a bisection algorithm in step $3$ of the Algorithm \ref{alg:Phi} (shown as circle-shaped points). By Theorem~\ref{thm:RT}, the intervals $I_d^1$ and $I_d^3$ (shown as bold lines) have crossing direction $+1$ (entering $\Cs$) since the sign of $\so\phi'(\w)$ is positive. The interval $I_d^2$ (shown as a dashed line) has crossing direction $-1$ (leaving $\Cs$) because the sign of $\so\phi'(\w)$ is negative.
\end{example}

In section~\ref{sec:whintvals} we compute the feasible intervals $I_h$ on the relative stability boundary by Algorithm~\ref{alg:H}. We divided the relative stability boundary into the intervals with invariant crossing directions $I_d$ by Algorithm~\ref{alg:Phi}. We obtain the feasible intervals with invariant crossing direction by intersecting these two sets of intervals, i.e., $I=\cup_{i=1}^{n_I} I_i=I_h \cap I_d$.

\begin{example}
In Table \ref{table:intI}, the first column shows the computed feasible intervals in Example~\ref{ex:Ih}. The second column contains the intervals with invariant crossing directions found in Example~\ref{ex:Id}. By intersection, we compute the intervals $I=\cup_{i=1}^{4} I_i$ where for each interval $I_i$, the crossing direction is invariant and $H(\w)$ is non-negative. The functions $H(\w)$ and $\phi(w)$ are monotonic on each interval $I_i$. The monotonicities of these functions are defined as the signs of $H'(\w)$ and $\phi'(\w)$, $\Phi_i$ and $H_i$ respectively. The last three columns in Table~\ref{table:intI} shows $\Phi_i$, $H_i$ and crossing direction for each interval $I_i$ which is equal to $-\Phi_i$ by Theorem~\ref{thm:RT} and $\so<0$.
\end{example}

\begin{table}[h]
  \tbl{Computation of the feasible intervals with invariant crossing direction.}
{\begin{tabular}{llllll}\toprule
  $I_h=\cup_{i=1}^3 I_h^i$ & $I_d=\cup_{i=1}^3 I_d^i$ & $I=\cup_{i=1}^4 I_i$ & $H_i$ & $\Phi_i$ & $\mathcal{CD}_i$ \vspace{1mm} \\
\colrule
  \rule{0pt}{3ex}
  $I_h^1=[0,1.144]$ & $I_d^1=[0,1.156]$ & $I_1=[0,1.144]$ & $+1$ & $-1$ & $+1$ (entering into $\Cs$)\vspace{1mm} \\
  $I_h^2=[1.144,1.369]$ & $I_d^2=[1.156,1.559]$ & $I_2=[1.144,1.156]$ & $-1$ & $-1$ & $+1$ (entering into $\Cs$)\vspace{1mm}\\
  $I_h^3=[2.249,\infty]$ & $I_d^3=[1.559,\infty]$ & $I_3=[1.156,1.369]$ & $-1$ & $+1$ & $-1$ (leaving $\Cs$)\vspace{1mm}\\
   &  &  $I_4=[2.249,\infty]$ & $+1$ & $-1$ & $+1$ (entering into $\Cs$)\vspace{1mm}\\
\botrule
\end{tabular}} \label{table:intI}
\end{table}

\subsubsection{Remarks on Assumptions}

Before computing the boundary crossing roots in the next section, we give some remarks on our assumptions. We assume that $G$ has no poles or zeros on the relative stability boundary $\Re(s)=\so$. Since $\so$ is a design parameter, a slight perturbation of $\so$ is sufficient to satisfy this assumption in numerical implementation.

Our second assumption is that there is no characteristic root on the relative stability boundary with multiplicity more than one, i.e., $f(s_0, h_0)=0$, $\frac{\partial f}{\partial s}(s_0, h_0)=0$. In general, this assumption always holds. Note that when the second assumption does not hold, the formula for the crossing direction (\ref{eq:RT}) is not valid. In this case the computation of the crossing direction requires a high-order analysis and recently there are on-going research on this direction \cite{FuACC07}, \cite{JarlebringTDW09}. In the following proposition, we give a necessary and sufficient condition to check the boundary crossing roots with multiplicity more than one.

\begin{proposition} \label{prop:multipleRoots}
The characteristic root $s_0=\sigma_0+j\w_0, s_0\in\C, \w_0\geq0$ of $G_{cl}$ has multiplicity more than one at $h=h_0$ if and only if there exists $\w_0\in\left\{\omega\geq0: H'(\w)=0\ \textrm{and}\ \phi'(\w)=0 \right\}$ such that
\be \label{eq:multipleset}
h_0=H(\w_0)\geq0\ \textrm{and}\ f(s_0,h_0)=0.
\ee
\end{proposition}
\begin{proof}
The conditions (\ref{eq:multipleset}) in the proposition guarantee that the characteristic root crosses $\Re(s)=\so$. The characteristic root of $G_{cl}$ with the multiplicity more than one holds
\be \label{eq:fsh0}
G'(s_0)G^{-1}(s_0)-h_0=0.
\ee
Substituting real and imaginary parts of $s_0$ into (\ref{eq:GpG}), the equation (\ref{eq:fsh0}) can be written as
\bd
\phi'(\w_0)+\w_0H'(\w_0)-jH'(\w_0)\sigma_0=0
\ed which is equivalent to
\bd
H'(\w_0)=0\ \textrm{and}\ \phi'(\w_0)=0.
\ed
The assertion follows.
\end{proof}

We choose $\so$ such that there are no poles or zeros of $G$ on the relative stability boundary $\Re(s)=\so$. We compute the nonnegative real zeros of $H'(\w)$ and $\phi'(\w)$ by Algorithm~\ref{alg:H} and \ref{alg:Phi}. We check whether the conditions in Proposition~\ref{prop:multipleRoots} hold for chosen $\so$. If it does, we slightly perturb $\so$ to satisfy the second assumption.
\vspace{-.6cm}
\section{Main Results} \label{sec:mainres}
The characteristic roots of $G_{cl}$ cross the relative stability boundary \emph{only} at critical delays. The number of characteristic roots inside $\Cs$ changes at each critical delay and it is constant inside the delay interval, an interval between two consecutive critical delays. We compute the number of characteristic roots inside $\Cs$ for each delay interval by first calculating the initial number of characteristic roots inside $\Cs$ (i.e., $h=0$) and updating this number at every critical delay.

Both methods, the relative stability analysis of closed-loop dead-time systems \emph{for a delay upper bound, $h_{\max}$} and \emph{for all delays}, are based on computation of critical delays in an increasing order with different stopping criteria. This computation is described in the next section and the relative stability analysis methods are explained in Section~\ref{sec:stabI}.

\subsection{Computation of Critical Delays in an Increasing Order} \label{sec:critdelays}
We first discuss to compute the characteristic roots crossing \emph{one} interval $I_i$ with increasing critical delays. This is necessary to compute the characteristic roots crossing the relative stability boundary in an increasing order until the final delay $h_{\max}$. Based on the computation on one interval, we develop a procedure to compute the boundary crossing roots with increasing delays over all the intervals in $I$.

By substituting (\ref{eq:tfG},\ref{eq:H}), we obtain the phase equation of (\ref{eq:complex}) on $\Re(s)=\so$ as
\be\label{eq:phiequation}
(2k+1)\pi-\phi_0=\phi(\w),\;k\in\Z
\ee where $\phi_0$ is the offset difference between the phase of $G_{cl}$ and $\phi(\w)$ (\ref{eq:phi}) defined as $\phi_0=\angle{G(\sigma_0)}-\phi(0)$.

Note that the left-hand side of the equation (\ref{eq:phiequation}) represents the horizontal lines \mbox{$(2l+1)\pi$} or \mbox{$2l\pi$} for $l\in\Z$ because $\phi_0$ is either $0$ or $\pi$. Each intersection of horizontal lines and $\phi(\w)$ over the intervals $I$ corresponds to a boundary crossing root since any point on $I$ holds the magnitude condition (\ref{eq:hw}) and the intersection point satisfies the phase condition of the characteristic equation (\ref{eq:phiequation}) on $\Re(s)=\so$.

On the interval $I_i$, the function $\phi(w)$ is monotonic. If a horizontal line intersects the function $\phi(w)$, we can compute the intersection point by a bisection algorithm for $\phi(\w)$ over the interval $I_i$. The intersection point is the imaginary part of the boundary crossing root on the interval $I_i$ and the corresponding critical delay is the value of $H(\w)$ for this point. If there is no horizontal line intersecting $\phi(\w)$ on $I_i$, we discard the interval since there is no characteristic root crossing this interval.

Both functions $\phi(w)$ and $H(\w)$ are monotonic on the interval $I_i$. Therefore the critical delays of intersections of horizontal lines and $\phi(\w)$ are also monotonic on $I_i$. Among the horizontal lines crossing $\phi(w)$ on $I_i$, the horizontal line resulting the smallest critical delay for $I_i=[\w_l^i,\w_r^i]$ is
\begin{itemize}
\item the smallest one if both $\phi(\w)$ and $H(\w)$ are increasing or decreasing ($\Phi_i H_i>0$) and it is equal to
$\phi_h^i=\left(2\left\lceil \frac{\phi(\w_l^i)+\phi_0}{2\pi}-\frac{1}{2} \right\rceil+1\right)\pi-\phi_0$,
\item the largest one if $\phi(\w)$ and $H(\w)$ have opposite monotonicity ($\Phi_i H_i<0$) and it is equal to
$\phi_h^i=\left(2\left\lfloor \frac{\phi(\w_r^i)+\phi_0}{2\pi}-\frac{1}{2} \right\rfloor+1\right)\pi-\phi_0$

\end{itemize} where $\Phi_i$ and $H_i$ are signs of $\phi'(\w)$ and $H'(\w)$ on the interval $I_i$. We find the point $\w^i$ such that $\phi_h^i=\phi(\w^i)$ by a bisection search. The smallest critical delay for the interval $I_i$ is the time-delay of this point, $h^i=H(\w^i)$. We can compute other critical delays for the same interval in an increasing order. Since the distance between consecutive horizontal lines is $2\pi$, we update the next horizontal line which is
\begin{itemize}
\item the one above $\phi_h^i \leftarrow \phi_h^i + 2\pi$ if $\Phi_i H_i>0$,
\item the one below  $\phi_h^i \leftarrow \phi_h^i - 2\pi$ if $\Phi_i H_i<0$
\end{itemize} and compute the intersection point $\w^i$ by a bisection search for the updated horizontal line and the corresponding next critical delay. The interval is discarded when the horizontal line is outside the interval range, i.e., $\phi_h^i \notin [\phi_{\min}^i,\phi_{\max}^i]$ where $\phi_{\min}^i$ and $\phi_{\max}^i$ are the minimum and maximum of $\phi(\w_r^i)$ and $\phi(\w_l^i)$ respectively.

\begin{example}
In Figure \ref{fig:I1}, the plots of $\phi(\w)$ and $H(\w)$ are given for the interval $I_1$ in Table \ref{table:intI}. The horizontal lines are the left-hand side of (\ref{eq:phiequation}). The functions $H(\w)$ and $\phi(\w)$ have opposite monotonicity ($\Phi_1 H_1<0$).  Among the horizontal lines intersecting $\phi(\w)$ on $I_1$, the largest one, $\phi_h^1=-\pi$, results the smallest critical delay. By a bisection search, the intersection point is $\w^1=0.642$ and its delay is $h^1=H(\w^1)=4.488$. Using the same approach for the updated horizontal line $\phi_h^1 \leftarrow \phi_h^1-2\pi=-3\pi$, the second critical delay on the interval $I_1$ is $h^1=9.209$ at $\w^1=1.031$.
\end{example}

\begin{figure}[!h]
   \begin{minipage}[t]{0.45\textwidth}
      \vspace{0pt}
      \includegraphics[width=\linewidth]{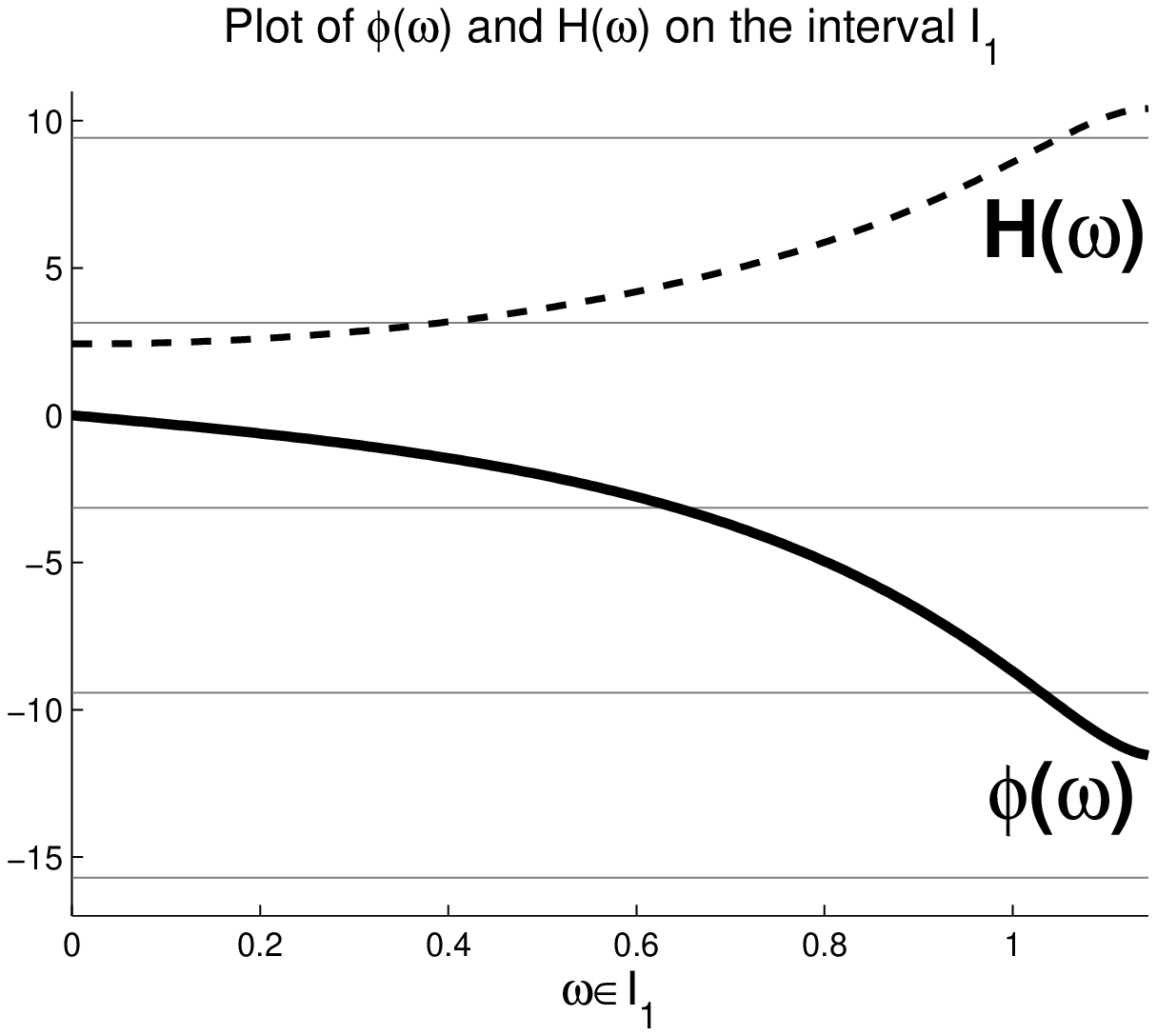}
      \caption{$\phi(\w)$ and $H(\w)$ plot on the interval $I_1$}
      \label{fig:I1}
   \end{minipage}
   \hfill
   \begin{minipage}[t]{0.45\textwidth}
      \vspace{0pt}\raggedright
      \includegraphics[width=\linewidth]{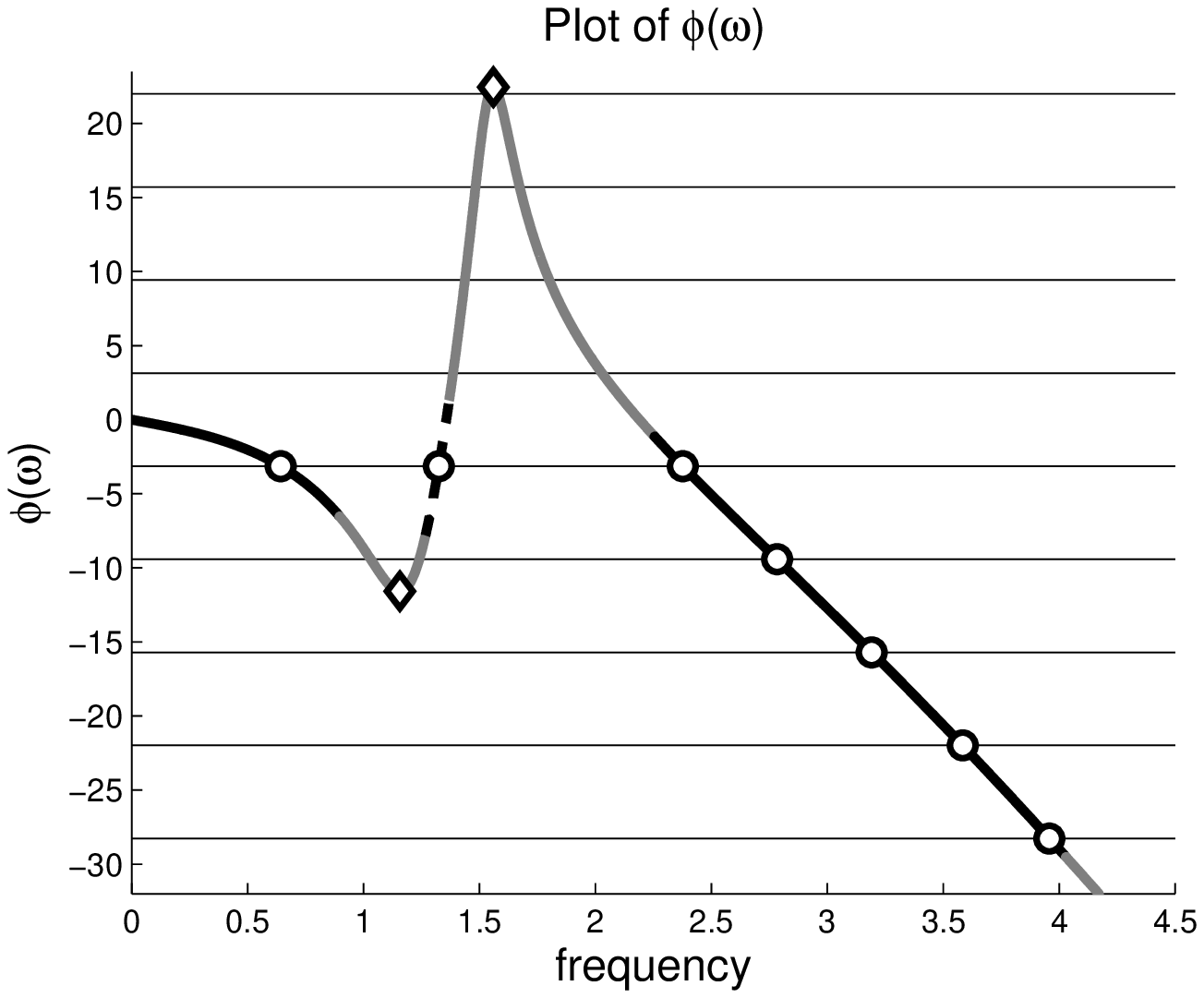}
      \caption{$\phi(\w)$ plot and the frequency intervals with invariant crossing direction}
      \label{fig:wc}
   \end{minipage}
\end{figure}

Based on the computation of critical delays in an increasing order for \emph{one} interval $I_i$, we can compute critical delays in an increasing order for all intervals in $I=\cup_{i=1}^{n_I} I_i$. This is done as follows. First we compute the smallest critical delay $h^i$ for each interval $I_i$ as explained above. The minimum of these delays $h^{i_{\min}}$ is the first critical delay of all intervals. Then at every iteration we compute the boundary crossing root with the next smallest delay \emph{only} in the interval $I_{i_{\min}}$ and update its delay $h^{i_{\min}}$. We find the next critical delay which is the minimum of all delays and repeat the iteration. If there is no boundary crossing root on an interval, we discard that interval.

The following algorithm computes critical delays in an increasing order for all intervals \mbox{$I=\cup_{i=1}^{n_I} I_i$}. Using boundary characteristic roots, it calculates the number of characteristic roots inside $\Cs$ between consecutive critical delays.

\begin{algorithm} \label{alg:critdelays}
\begin{enumerate}
\item Compute ${\bf n^1}$, the initial number of characteristic roots of $G_{cl}$ inside $\Cs$ (the zeros of $G_{cl}(0)$ inside $\Cs$) and set the first critical delay to zero, ${\bf h_c^1=0}$.
\item For each interval $I_i$ in $I$,
\begin{itemize}
\item calculate the horizontal line $\phi_h^i$ depending on the sign of $\Phi_i H_i$,
\item if the horizontal line $\phi_h^i$ is outside of the range $[\phi_{\min}^i,\phi_{\max}^i]$, discard this interval,
\item compute the intersection point $\w^i$ of $\phi_h^i$ and $\phi(\w)$ by a bisection algorithm,
\item calculate the smallest critical delay of \emph{this} interval, $h^i=H(\w^i)$.
\end{itemize}
\item Repeat for ${\bf k=1,2,\ldots}$
\begin{enumerate}
\item Find the minimum delay $h^{i_{\min}}$ of delays $h^i$ for remaining intervals. The next critical delay is ${\bf h_c^{k+1}}=h^{i_{\min}}$ and the number of characteristic roots inside $\Cs$ for delay interval $[h_c^{k},h_c^{k+1}]$ is $n^k$.
\item Update the number of characteristic roots inside $\Cs$ for the next delay interval,
\bd
n^{k+1}=n^{k}+\left\{
      \begin{array}{rr}
        2 \\
        1 \\
      \end{array} \times \mathcal{CD}_i\ \rm{for}\
      \begin{array}{r}
        w^{i_{\min}}\neq0; \\
        w^{i_{\min}}=0; \\
      \end{array}
    \right.
\ed
where the formula counts the real or complex boundary crossing roots according to their crossing direction, $\mathcal{CD}_i=-\Phi_{i_{\min}}$ at the interval $I_{i_{\min}}$.
\item Apply Step $2$ \emph{only} for the interval $I_{i_{\min}}$. Update $h^i$ and $\w^i$ by computing the next smallest critical delay and the imaginary part of its characteristic root for \emph{this} interval or discard the interval if $\phi_h^{i_{\min}}$ is outside of the range $[\phi_{\min}^{i_{\min}},\phi_{\max}^{i_{\min}}]$.
\end{enumerate}
\end{enumerate}
\end{algorithm}

There are infinitely many critical delays crossing the relative stability boundary \mbox{$\Re(s)=\so<0$}, \cite{WimBook}. Our numerical methods for the considered problems in Section~\ref{sec:problem} are based on Algorithm~\ref{alg:critdelays} with different stopping criteria which are described in the following section.

\subsection{Relative Stability Analysis of the Closed-Loop SISO Dead-Time Systems} \label{sec:stabI}

The numerical method for the first formulated problem in Section~\ref{sec:problem} computes delay intervals and the number of characteristic roots of $G_{cl}$ inside $\Cs$ at each delay interval for $h\in[0,h_{\max}]$. In Algorithm~\ref{alg:critdelays} we continue to compute critical delays in an increasing order until \emph{the critical delay is larger than $h_{\max}$}. Delay intervals are the intervals between critical delays and the end point of the last delay interval is set to $h_{\max}$. The algorithm also provides the number of characteristic roots inside $\Cs$ for each delay interval.

The second problem is the computation of the delay intervals where $G_{cl}$ is stable on the region $\Re(s)<\so$ for $h\in[0,\infty)$. In Algorithm~\ref{alg:critdelays} we stop when \emph{the number of characteristic roots inside $\Cs$ is larger than the number of characteristic roots leaving $\Cs$ after the last computed critical delay}. When this condition holds, there are always characteristic roots inside $\Cs$ and $G_{cl}$ is unstable for all delays larger than the last critical delay. The number of characteristic roots leaving $\Cs$ for all delays is computed in the following proposition.
\begin{proposition} \label{prop:leaving}
Given the relative stability boundary $\Re(s)=\so<0$, the number of characteristic roots leaving $\Cs$ for $h\in[0,\infty)$ is finite and equal to the summation of terms
\be
\left\lfloor \frac{\phi_{\max}^i+\phi_0}{2\pi}-\frac{1}{2} \right\rfloor - \left\lceil \frac{\phi_{\min}^i+\phi_0}{2\pi}-\frac{1}{2} \right\rceil +1
\ee for all intervals $I_i$ with crossing direction $-1$.
\end{proposition}
\begin{proof}
The number of characteristic roots leaving $\Cs$ for $h\in[0,\infty)$ is equal to the number of intersections of horizontal lines on the left-hand side of equation (\ref{eq:phiequation}) and the function $\phi(\w)$ over the intervals $I_i$ with the crossing direction $-1$. Since $\phi(\w)$ is a continuous function, it is sufficient to show that all intervals with the crossing direction $-1$ are finitely many and have a finite length.

By Corollary~\ref{cor:extremumpts}, the intervals whose boundary points are the non-negative real zeros of $H'(\w)$, $0$ and $\infty$ contain the intervals in $I$. As a result, the number of intervals in $I$ is finite and all intervals have finite length except the last one. Therefore it is enough to show that either the last interval $I_{n_I}$ has a finite length or its crossing direction is not $-1$. When $|d|\leq1$ where $d:=G(\infty)$, the last interval $I_{n_I}$ has infinite length, $\lim_{\w\rightarrow\infty}H(\w)=\frac{\ln|d|}{\so}\geq0$. However, the crossing direction of this interval is not $-1$ since \mbox{$\lim_{\w\rightarrow\infty} \so\phi'(\w)= \ln\frac{1}{|d|}\geq0$}. When $|d|>1$, the last interval $I_i$ has a finite length since $\lim_{\w\rightarrow\infty}H(\w)=\frac{\ln|d|}{\so}<0$ does not satisfy (\ref{eq:hw}). The result follows. The formula in the proposition sums the number of horizontal lines intersecting $\phi(\w)$ over the intervals in $I$ with crossing direction $-1$.
\end{proof}

The numerical method for the second problem first computes ${\bf n_s}$, the number of characteristic roots leaving $\Cs$ for all delays by Proposition~\ref{prop:leaving}. In step $3-b)$ of Algorithm~\ref{alg:critdelays}, $n_s$ is updated if the computed boundary crossing root leaves $\Cs$, i.e., its crossing direction is $-1$. We stop the algorithm when \emph{$n^k$ is larger than $n_s$}. The closed-loop system $G_{c}$ is stable where delay intervals have no characteristic roots inside $\Cs$, i.e., $n^k=0$.

\begin{remark} When the plant $G$ is bi-proper ($d\neq0$) and $|d| \geq 1$, the closed-loop system $G_{cl}$ has infinitely many characteristic roots inside $\Cs$ for all delays.\end{remark}
\begin{remark}
When the plant $G$ is bi-proper and $|d|<1$ , both methods compute critical delays up to $\frac{\ln |d|}{\so}$. For greater delays, the closed-loop system $G_{cl}$ is unstable and has infinitely many characteristic roots inside $\Cs$. \end{remark}
\begin{remark}
 Our numerical method is applicable when the stability boundary is chosen as the imaginary axis, i.e., $\so=0$. Since the stability analysis of time-delay systems for this case is well-studied, \cite{WimBook}, we give few remarks on the implementation of the numerical method. In this case, feasible intervals reduce into \emph{feasible frequency points} including \emph{critical frequencies} satisfying
\be \label{eq:imcase}
H_{im}(\w):=\ln|G(j\w)|=0.
\ee
These points are calculated using a bisection algorithm and the nonnegative real zeros of $H_{im}'(\w)$ similar to the nonimaginary axis case. Then the critical frequencies are the points satisfying (\ref{eq:imcase}) and the characteristic equation (\ref{eq:complex}). The corresponding delay values for a critical frequency $\w_0$ are $H_{im}(\w_0)+\frac{2k\pi}{\w_0}, k\in\Z$ and its crossing direction is equal to
\bd
\mathcal{CD}(\so+j\w_0,H_{im}(\w_0))=\textrm{sign}\left(-H_{im}'(\w_0)\right)
\ed which is the special case of Theorem \ref{thm:RT} for the imaginary axis case. The stability of the closed-loop SISO dead-time system can be determined based on the computed critical frequencies and using the periodicity of the corresponding delays.
\end{remark}

\section{Examples} \label{sec:examples}
The characteristic roots of $G_{cl}$ given $G$ (\ref{eq:ex1}) cross the relative stability boundary $\Re(s)=\sigma_0=-0.1$ at seven critical delays for $h\in[0,7]$ as shown circle-shaped dots in Figure~\ref{fig:wc}. Delay intervals, the number of characteristic roots inside $\Cs$ in these intervals and the boundary crossing roots are given in Table~\ref{table:fourhc}.

\begin{table}[h]
\tbl{Delay intervals, the number of characteristic roots inside $\Cs$ for each delay interval
and boundary crossing roots for \mbox{$h\in[0,7]$} and \hbox{$\Re(s)=-0.1$}.}
{\begin{tabular}{rcc}\toprule
  Delay  & \# of Characteristic  & Boundary Crossing \\
  Intervals & Roots inside $\Cs$ & Roots \\
\colrule
  \rule{0pt}{3ex}
  $[0,0.879)$ & $0$ & $-$\vspace{1mm} \\
  $(0.879,2.984)$ & $2$ &  $-0.1+2.377j$\vspace{1mm}\\
  $(2.984,3.280)$ & $4$ &  $-0.1+2.784j$\vspace{1mm}\\
  $(3.280,4.488)$ & $2$ &  $-0.1+1.325j$\vspace{1mm}\\
  $(4.488,4.556)$ & $4$ &  $-0.1+0.642j$\vspace{1mm}\\
  $(4.556,5.800)$ & $6$ &  $-0.1+3.192j$\vspace{1mm}\\
  $(5.800,6.831)$ & $8$ &  $-0.1+3.584j$\vspace{1mm}\\
  $(6.831,7]$  & $10$ &  $-0.1+3.958j$\vspace{1mm}\\
\botrule
\end{tabular}}  \label{table:fourhc}
\end{table}
In the second example, $G$ is a high-order approximation of the one dimensional heat diffusion equation, \cite{CurtainAUT2009}
\bd
G(s)=\prod_{n=1}^{100} \frac{1+\frac{s}{n^2\pi^2}}{1+\frac{s}{(n-1/2)^2\pi^2}}
\ed The relative stability analysis of the closed-loop of this system with an input delay for various $\so$ is given in Table \ref{table:highorderex}. The results are obtained within $5$ sec. for each  $\so$ on a PC with an Intel Core Duo 2.53 GHz processor with 2 GB RAM.

\begin{table}[h]
\tbl{Delay intervals where $G_{cl}$ is stable.}
{\begin{tabular}{cccc} \toprule
  $\sigma_0$ & $-0.1$ & $-0.5$ & $-1$\\
  stable delay intervals & $[0,1.575)$ & $[0,0.770)$ & $[0,0.551)$ \\
\botrule
\end{tabular} \label{table:highorderex}}
\end{table}

\begin{table}[h]
\tbl{Delay intervals where $G_{cl}$ is stable.}
{\begin{tabular}{cccccc} \toprule
  Reference & $G(s)$ &  & Delay Intervals &  & \\
  \toprule
  \cite{Thowsen1981} &  $\frac{1}{s^3+s^2+2s+1}$ & $\so=0$ & $\so=-0.01$ & $\so=-0.02$ & $\so=-0.03$ \\
   &   &
   $\begin{array}{c}
     (\frac{\pi}{2},\sqrt{2}\pi) \\
     (\frac{5\pi}{2},2\sqrt{2}\pi)
   \end{array}$
   & $(1.714,4.267)$ & $(1.878,4.125)$ & $(2.098,3.894)$ \\
  \colrule
  \cite{ChenSCL1995} & $\frac{s}{s^2+s+1}$ & $\so=0$ & $\so=-0.01$ & $\so=-0.1$ & $\so=-0.5$ \\
   & & $\begin{array}{c}
         [0,\infty)\ \textrm{except}\\
         (2k+1)\pi,\ k\in\Z
       \end{array}$  & $\begin{array}{c}
                          [0,2.467) \\
                          (4.209,7.261)
                        \end{array}$ & $[0,1.612)$ & $[0,0.811)$ \\
  \colrule
  \cite{Louisell2001} & $\frac{-(s+2)}{s^2+s+4}$ & $\so=0$ & $\so=-0.01$ & $\so=-0.1$ & $\so=-0.5$ \\
   & & $\begin{array}{c}
         [0,2.006)\\
         (4.443,4.571)
       \end{array}$  & $(0.010,1.971)$ & $(0.105,1.745)$ & $(0.573,1.311)$ \\
  \colrule
  \cite{HanYuGu2004} & $\begin{array}{c}
         \frac{b-cs}{s}\ \textrm{for}\\
         (b,c)=(0.1,3)
       \end{array}$ $$  & $\so=0$ & $\so=-0.01$ & $\so=-0.1$ & $\so=-1$ \\
   & & $\begin{array}{c}
         [0,0.488)
       \end{array}$  & $[0,0.484)$ & $[0,0.453)$ & $[0,0.294)$ \\
  \colrule
  \cite{HuLiu2007} & $\frac{-0.2s+1}{s}$ & $\so=0$ & $\so=-0.01$ & $\so=-0.5$ & $\so=-1$ \\
   & & $[0,1.342)$  & $[0,1.309)$ & $[0,0.655)$ & $[0,0.452)$ \\
\botrule
\end{tabular} \label{table:variousexamples}}
\end{table}

We collected various examples from the literature and give our results in Table \ref{table:variousexamples}. The first three columns show the references, plants and delay intervals where the closed-loop system $G_{cl}$ is stable given in the corresponding papers. The next three columns are delay intervals where the closed-loop system is relative stable for various stability boundaries different than the imaginary axis. Note that in practice, the relative stability is required for the robustness of the overall system and delay intervals with respect to the imaginary axis and the relative stability boundary are considerably different even if they are arbitrarily close, i.e., $\so=-0.01$.

For a relative stable closed-loop system, the real part of the dominant pole $\sigma$ determines the rate of convergence of system states and it is proportional to $e^{\sigma t}$. The settling time of system states corresponding for $\pm0.2\%$ tolerance is approximated as $\frac{4}{|\sigma|}$, \cite{OgataBook}. When the closed-loop of a SISO dead-time system is relative stable with respect to the boundary $\Re(s)=\so$ for a delay interval, all poles of characteristic roots have real parts less than $\so$ for all delays in this interval. Therefore the upper bound of the settling time for all delays in the delay interval is $\frac{4}{|\so|}$.

\begin{figure}[!h]
   \begin{minipage}[t]{0.45\textwidth}
      \vspace{0pt}
      \includegraphics[width=\linewidth]{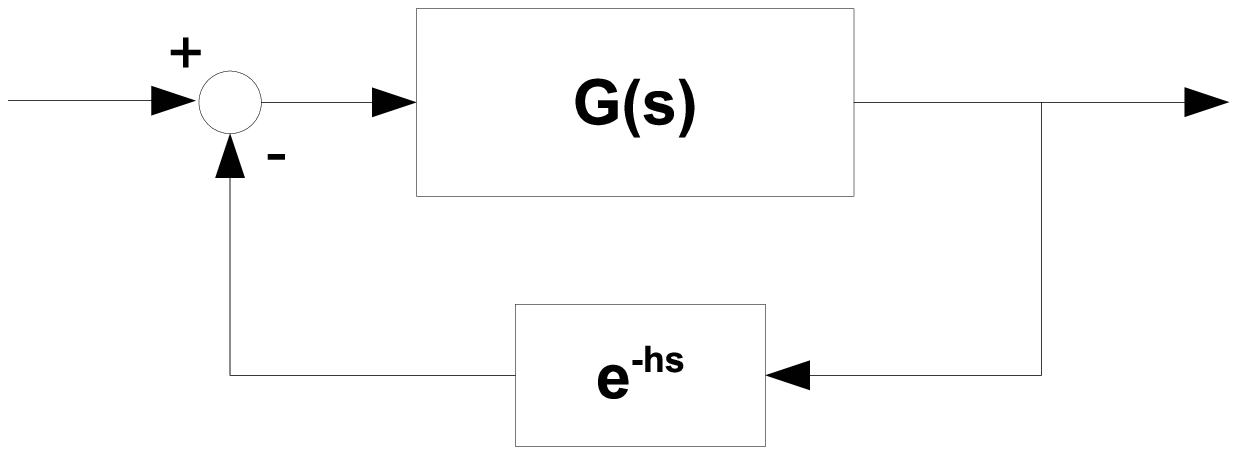}
      \caption{$The feedback configuration$}
      \label{fig:feedback}
   \end{minipage}
   \hfill
   \begin{minipage}[t]{0.45\textwidth}
      \vspace{0pt}\raggedright
      \includegraphics[width=\linewidth]{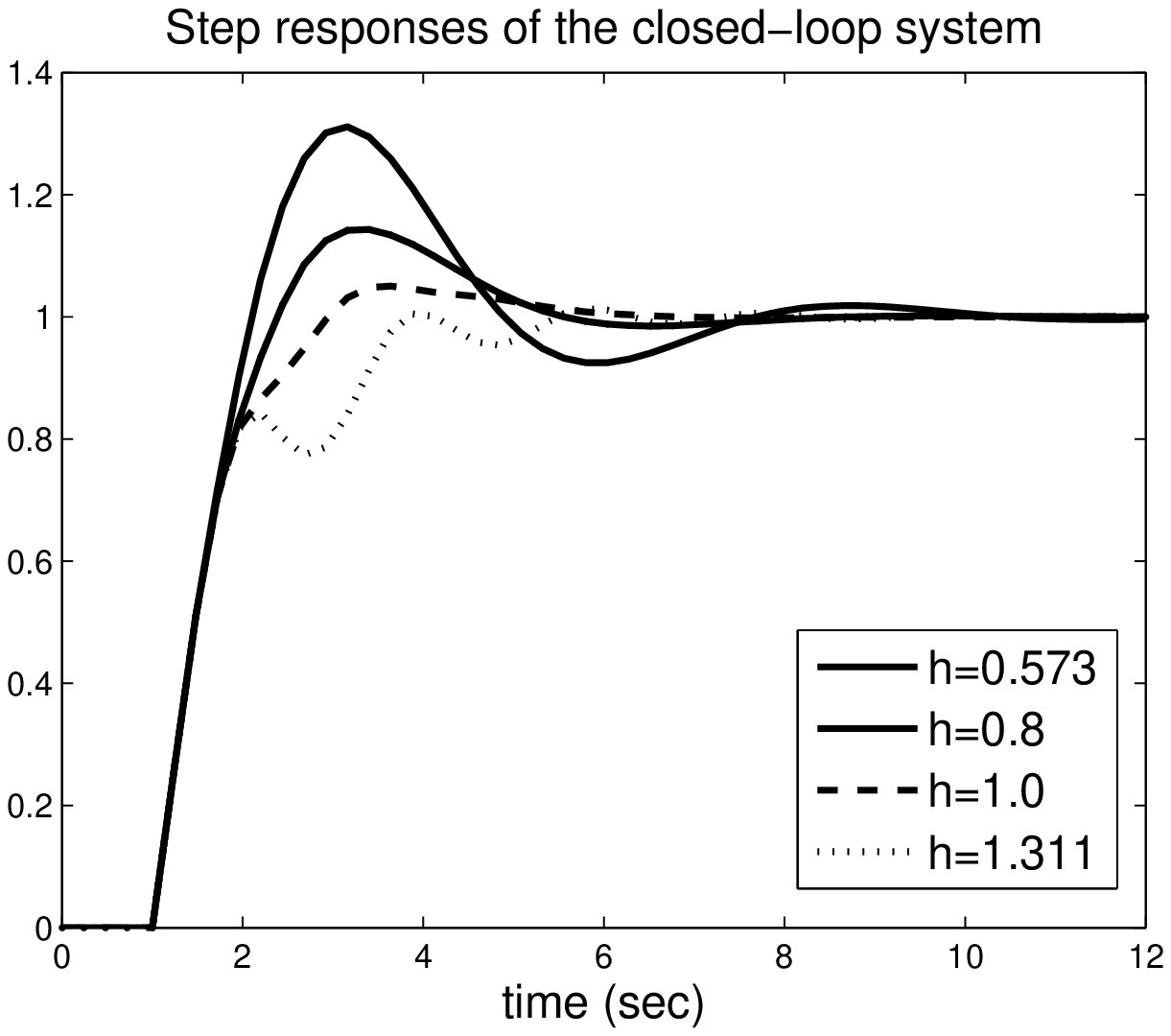}
      \caption{Step responses of the closed-loop system for $h=0.573,\ 0.8,\ 1.0,\ 1.311$.}
      \label{fig:stepresponses}
   \end{minipage}
\end{figure}

We consider the example in \cite{Louisell2001} for the feedback configuration in Figure \ref{fig:feedback} and the relative stability boundary is $\so=-0.5$. The characteristic roots of the closed-loop system $G_{cl}$ is relative stable for $h\in(0.573,1.311)$. The responses to a negative step input of the feedback system for $h=0.573,\ 0.8,\ 1.0,\ 1.311$ are given Figure \ref{fig:stepresponses}. Note that all responses have settling time less than $\frac{4}{0.5}=8$ seconds as expected. We also verified that settling times are larger outside of this delay interval.

\section{Concluding Remarks} \label{sec:concl}
We present a method to compute the critical delays in an increasing order on a given vertical line stability boundary different than the imaginary axis for the closed-loop SISO dead-time systems. Based on this method, we calculate the characteristic roots crossing the relative stability boundary and their critical delays. We analyze the relative stability of the closed-loop SISO dead-time system for a given time-delay interval or for all delays.

Our approach for the relative stability analysis of the closed-loop SISO dead-time systems can be extended to piece-wise linear stability boundaries. This allows analyzing the stability of such systems on complicated stability regions known as D-stability regions in the literature, \cite{Mao2006}. Available methods give sufficient conditions based on Lyapunov techniques and our future research direction is to analyze the D-stability of time-delay systems directly.

\section*{Acknowledgement(s)}
This article present results of the Belgian
Programme on Interuniversity Poles of Attraction, initiated by the
Belgian State, Prime Minister's Office for Science,
Technology and Culture, the Optimization in
Engineering Centre OPTEC of the K.U.Leuven, and the project STRT1-09/33
of the K.U.Leuven Research Foundation.

\bibliographystyle{tCON}
\bibliography{main}
\end{document}